\documentclass[final,authoryear,10pt]{elsarticle2}

\usepackage{amssymb}
\usepackage[hyphens]{url}
\usepackage[breaklinks=true]{hyperref}

\usepackage{lineno}

\usepackage{ctable}
\usepackage{tabularx}
\newcolumntype{+}{>{\global \let \currentrowstyle \relax}}
\newcolumntype{^}{>{\currentrowstyle }}

\usepackage[cmex10]{amsmath}

\hypersetup{linktocpage=true, urlcolor=magenta,linkcolor=blue,citecolor=blue, colorlinks=true,bookmarks=true,bookmarksopen=true,bookmarksopenlevel=0}
\newtheorem{thm}{Theorem}[section]

\newtheorem{rem}[thm]{Remark}

\newenvironment{proof}{%
{\noindent \bf Proof. }%
}{%
\hfill$\Box$\\%
}

\def\R{\mathbb{R}}

\def \e{{\rm e}}

\def \G {\mathcal{G}}
\def \K {\mathcal{K}}
\def \x {{\bf x}}

\usepackage{lscape}

\begin{document}

\begin{frontmatter}



\title{
A  reaction-diffusion system to better comprehend the unlockdown: Application  of  SEIR-type model with diffusion to the spatial spread of COVID-19 in France}


\author[lamfa]{Youcef Mammeri\corref{cor1}}
\address[lamfa]{Laboratoire Ami\'enois de Math\'ematique Fondamentale et Appliqu\'ee, CNRS UMR 7352, Universit\'e de Picardie Jules Verne, 80069 Amiens, France}

\cortext[cor1]{Corresponding author: youcef.mammeri@u-picardie.fr}

\begin{abstract}
A reaction-diffusion model was developed describing the spread of the COVID-19 virus  considering the mean daily movement of susceptible, exposed and asymptomatic individuals. The model was calibrated using data on the confirmed infection and death from France as well as their initial spatial  distribution. First, the system of partial differential equations is studied, then the basic reproduction number, $\mathcal{R}_0$ is derived. 
Second, numerical simulations, based on a combination of level-set and finite differences, shown the spatial spread of COVID-19 from March 16 to June 16. Finally,  scenarios of unlockdown are compared according to variation of distancing, or partially spatial lockdown.
\end{abstract}

\begin{keyword}
    COVID-19 \sep Reaction-diffusion \sep SE$I_s$$I_a$UR model \sep Reproduction number s\sep Unlockdown map

\MSC[2010] 92D30\sep 37N25 \sep 35K51 \sep 35Q92
\end{keyword}

\end{frontmatter}

\section{Introduction}
In late 2019, a disease outbreak emerged in the city of Wuhan, China.  The culprit was a certain strain called  Coronavirus Disease 2019 or COVID-19 in brief \citep{WHO2020a}.  This virus has been identified to cause fever, cough, shortness of breath, muscle ache, confusion, headache, sore throat, rhinorrhoea, chest pain, and nausea \citep{Hui2020, Chen2020}. COVID-19 belongs to the \textit{Coronaviridae} family. A family of coronaviruses that cause diseases in humans and animals, ranging from the common cold to severe diseases.  Although only seven coronaviruses are known to cause disease in humans,  three of these,  COVID-19 included, can cause severe infection, and sometimes fatal to humans. 

COVID-19 spreads fast. According to WHO \citep{WHO2020b}, it only took  67 days from the beginning of the outbreak in China last December 2019 for the virus to infect the first 100,000 people worldwide.
As of the 5th of May 2020, a cumulative total of 3,601,760 confirmed cases, while 251,910 deaths have been recorded for COVID-19 by World Health Organization \citep{WHO2020c}.

Last 30th of January, WHO characterized COVID-19 as Public Health Emergency of International Concern (PHEI) and urge countries to put in place strong measures to detect disease early, isolate and treat cases, trace contacts, and promote social distancing measures commensurate with the risk \citep{WHO2020d}. In response, the world implemented its actions to reduce the spread of the virus. Limitations on mobility, social distancing, and self-quarantine have been applied.  More than 100 countries established full or partial lockdown.   All these efforts have been made to reduce the transmission rate of the virus. For the time being,  COVID-19 infection is still on the rise. 

Many mathematical models have been proposed to help governments as an early warning device about the size of the outbreak, how quickly it will spread, and how effective control measures may be. Most of the model are  (discrete or continuous) SIR-type and few  are taken into account the spatial spread.

\cite{Gardner2020}  implemented a metapopulation network model described  by a  discrete-time Susceptible-Exposed-Infected-Recovered (SEIR) compartmental model. The model gives an estimate of the expected number of cases in mainland China at the end of January 2020, as well as the global distribution of infected travelers. 
\cite{Wu2020} fused an SEIR metapopulation model to simulate epidemic.  \cite{leon20}  incorporated daily movements in an SEIR model, while  \cite{giu20} proposed a statistical model to handle the diffusion of covid-19 in Italy.

Here, spatial propagation is translated by a diffusion and the reaction terms are deduced from an extension from the classical SEIR model by adding a compartment of asymptomatic infected \citep{arc20}.  We aim at predicting the spread of COVID-19 by giving maps the basic reproductive numbers $\mathcal{R}_0$ and its effective reproductive number $\mathcal{R}_{\rm eff}$.  Afterward, we also assess possible scenarios of unlockdown.

The rest of the paper is organized as follows. Section 2, outlines our methodology. Here the model was explained, where the data was taken, and its parameter estimates.  Section 3 contained the qualitative analysis for the model. Here, we provide the reproductive number $\mathcal{R}_0$, then compare strategies to handle unlockdown.  Finally, section 4 outlines our brief discussion on some measures to limit the outbreak.

\section{Materials and methods}

\subsection{Confirmed and death data}
In this study, we used the publicly available dataset of COVID-19 provided by the Sant\'e Publique France.
This dataset includes daily count of confirmed infected cases, recovered cases, hospitalizations and deaths. Data can be downloaded from 
{\footnotesize \url{https://www.data.gouv.fr/fr/datasets/donnees-hospitalieres-relatives-a-lepidemie-de-covid-19/}.}  
These data are collected by the National Health Agency and are directly reported public and unidentified patient data, 
so ethical approval is not required.
The map of population density are from G\'eoportail  ({\footnotesize \url{https://www.geoportail.gouv.fr}}) established by the National Geographic Institute. Data concerning transport are extracted from National Institute of Statistics and Economic Studies 
({\footnotesize \url{ https://statistiques-locales.insee.fr/}}).

\subsection{Mathematical model}

\begin{figure}[htbp]
\centering
\includegraphics[scale=0.5]{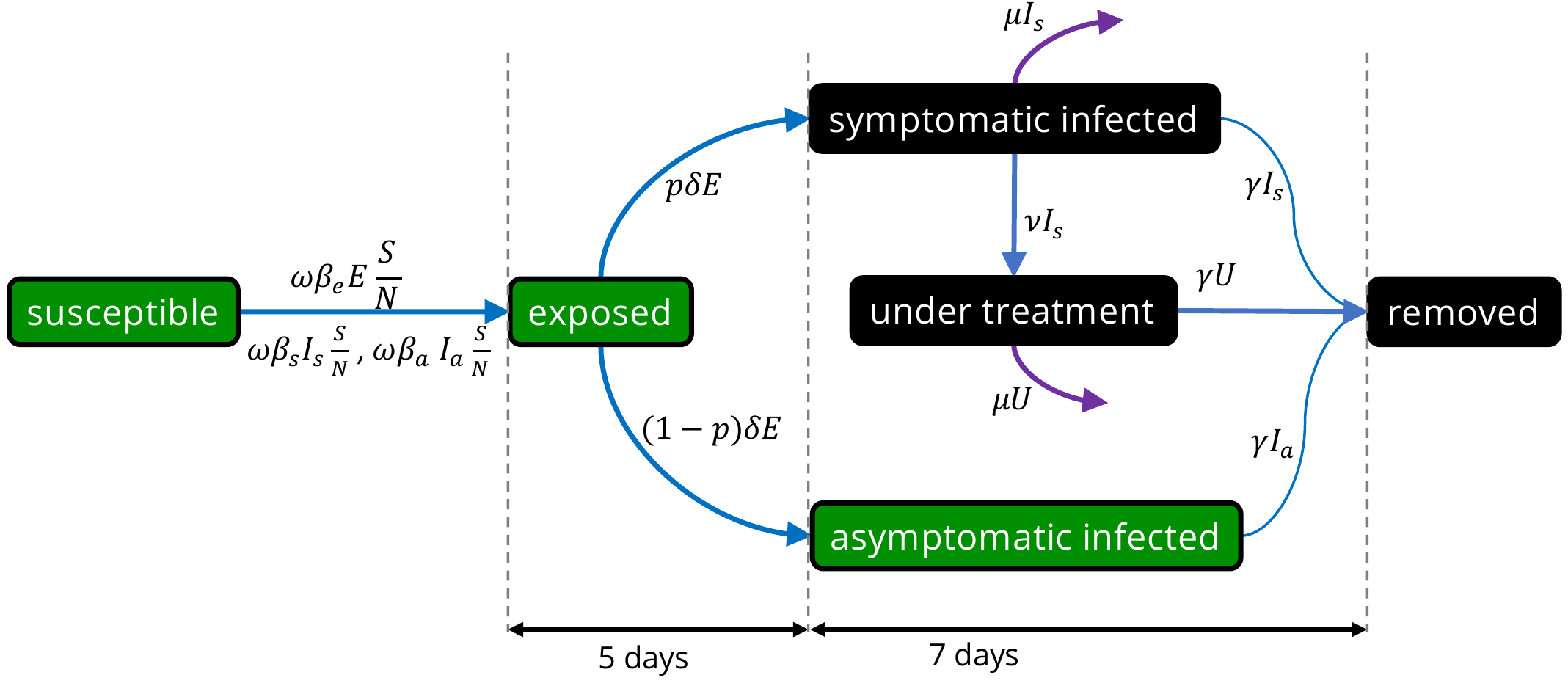} 
\\[-0.5cm]\caption{Compartmental representation of the $SEI_aI_sUR-$model. 
Blue arrows represent the infection flow. Purple arrows denote the death. Green compartments indicate moving individuals.}
\label{figflow}
\end{figure}

We focus our study on six components of the epidemic flow (Figure \ref{figflow}), {\it i.e.} the densities of Susceptible individual ($S$),
Exposed individual ($E$), symptomatic Infected individual ($I_s$), asymptomatic Infected individual ($I_a$), Under treatment
 individual ($U$) and Removed individual ($R$). To simplify the readings, treatments are not distinguished 
 between quarantine, hospitalization or medicine. 
To build the mathematical model, we followed the standard strategy developed in the literature concerning SIR model \citep{odi2000,bra2012,arc20}. We assumed that susceptible can be infected by exposed and by infected individuals. 
We suppose that only susceptible, exposed and asymptotic individuals are moving. 
The dynamics is governed by a system of three partial differential equations (PDE)
and three ordinary differential equations (ODE) as follows, for $\x=(x,y) \in \Omega \subset \R^2$, $t>0$,

\begin{equation}
  \label{e1}
  \left\{
      \begin{aligned}
     \partial_t S - d(t) \Delta S &= -   \omega(t)\left( \beta_e E + \beta_s I_s + \beta_a I_a \right) \frac{S}{N} 
    \\
\partial_t E - d(t) \Delta E &=     \omega(t)\left( \beta_e E + \beta_s I_s + \beta_a I_a \right) \frac{S}{N} - \delta E 
    \\
\partial_t I_a - d(t) \Delta I_a &= (1-p) \delta E - \gamma I_a
  \\
  I_s' &= p \delta E - (\gamma + \mu + \nu) I_s
  \\
  U' &= \nu I_s  - (\gamma + \mu) U 
    \\
  R' &= \gamma (I_a + I_s + U).   
      \end{aligned}
    \right.
\end{equation}
Frontiers being now closed, homogeneous Neumann boundary conditions is imposed.
The total living population  is $ N =  S + E + I_a + I_s + U + R$
and the death is $D = \mu (I_s+  U )$ 
No new recruit is added. 
The parameters are described in Table A of Figure \ref{fig1}.

\subsection{Parameters estimation}

Latency period and infection period have been estimated  as 5 days and 7 days respectively \citep{lau20}, and thus
$\delta=1/5, \gamma=1/7$. To account for the lockdown and unlockdown, the average number of contacts
is updated as follows \citep{liu20}
\begin{equation}
  \label{e2}
  \omega(t) = \left\{
      \begin{aligned}
        \omega_0 & \mbox{ if } t \leq t_{bol} \\
        \omega_0 \e^{-\rho (t- t_{bol})}& \mbox{ if } t_{bol} \leq t \leq t_{eol}    \\
              \frac{(1-\eta) \omega_0}{1+ ((1-\eta)\e^{\rho (t_{eol}- t_{bol})}-1) \e^{-2\rho (t- t_{eol})} } & \mbox{ if }  t \geq t_{eol},             
              \end{aligned}
    \right.
\end{equation}
while the diffusion coefficient is set up to
\begin{equation}
  \label{e3}
  d(t) = \left\{
      \begin{aligned}
        d_0 & \mbox{ if } t \leq t_{bol} \\
        d_0 \e^{-\rho (t- t_{bol})} & \mbox{ if } t_{bol} \leq t \leq t_{eol}    \\
      \frac{d_0}{1+ (\e^{\rho (t_{eol}- t_{bol})}-1) \e^{-2\rho (t- t_{eol})} } & \mbox{ if }  t \geq t_{eol}.            
         \end{aligned}
    \right.
\end{equation}
Here $bol$ denotes for beginning of lockdown and $eol$ for end of lockdown. Unlockdown is assumed to be faster
than lockdown. The parameter $0\leq \eta \leq 1$ is a varying coefficient translating
respect for distancing.
From INSEE, the daily commuting in France is around $25 km$, the value of $d_0$ is fixed equal to $\frac{25^2}{16}$ \citep{O80, SK97}.
Six parameters $\theta = (\rho,\beta_e,\beta_s,\beta_a,p,\mu)$ remain to be determined. Given, for $N$ days, the observations 
$I_{s, obs}(t_i)$ and $D_{obs}(t_i)$, the cost function consists of the nonlinear least square function
$$
J(\theta) = \sum_{i=1}^N  \left(I_{s, obs}(t_i) - \overline{I_s}(t_i,\theta) \right)^2 +  ( D_{obs}(t_i) - \overline{D}(t_i,\theta))^2,
$$
with constraints $\theta \geq 0$.
Here $\overline{I_s}(t_i,\theta) = \int_\Omega I_s(\x, t_i,\theta)d\x$ and $\overline{D}(t_i,\theta) = \int_\Omega D(\x, t_i,\theta)d\x$ denote the output of the mathematical model at time $t_i$ computed with the parameters $\theta$.  The optimization problem is solved using Approximate Bayesian Computation combined with a quasi-Newton method \citep{Csi2010}.

\begin{landscape}
\begin{figure}[htbp]
\centering
 \includegraphics[scale=.6]{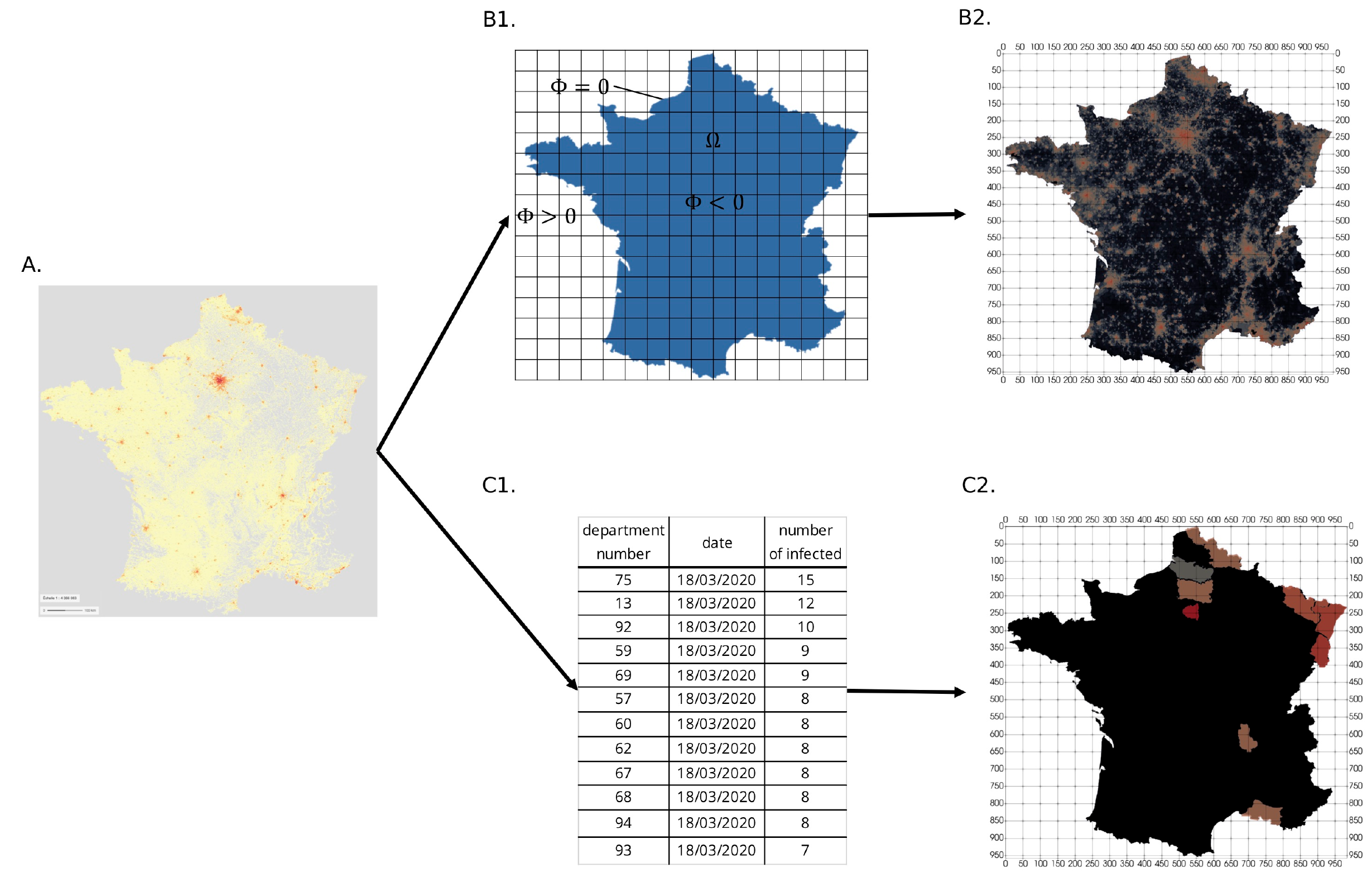}
 \\[-0.5cm]\caption{A. Map of population density from G\'eoportail.  B1. Level-set and  computation domain consisting in the cartesian grid. 
 B2. Initial population density  $N_0(\x)$. C1. Confirmed infecteds with respect to department number provided by Sant\'e Publique France. 
 C2. Initial infection density  $I_{s,0}(\x)$.}
\label{fig0}
\end{figure}

\end{landscape}

\subsection{Numerical discretization}

From the map of the country, the level-set  \citep{OF02, S99} is defined 
as the function $\phi$ such that the  territory 
$$  \Omega := \left\{ \x \in \R^2; \; \phi(\x) < 0 \right\}, $$ and its boundary is the zero level of $\phi$ 
The exterior normal is
$ \vec{n} = \frac{\nabla \phi}{||\nabla \phi ||}.$
The computation domain consists in a cartesian grid given by the image pixels
(Figure \ref{fig0}).
Then the map of population density allows to build initial population density  $N_0(\x)$.
From data of confirmed infecteds, the spatial distribution  of  initial  infection  $I_{s,0}(\x)$ with respect to department number
is created (Figure \ref{fig0}).

Finally, Runge-Kutta 4 is used for the time discretization and central finite difference for the space discretization.

\section{Results}

\subsection{Global well-posedness and basic reproduction numbers}

Let us suppose that for all time, $d_1 \leq d(t) \leq d_0$.
We first prove that the model is globally well posed.
\begin{thm}\label{thm1}
  Let $0 \leq S_0, E_0, I_{a,0}, I_{s,0}, U_0, R_0  \leq N_0$ be the initial datum. 
Then there exists a unique global in time weak solution 
$(S, E, I_{a}, I_{s}, U, R) \in L^\infty(\R_+,L^\infty(\Omega))^6,$ 
of the initial boundary value problem. Moreover, the solution is nonnegative and 
$S + E+ I_{a}+ I_{s}+ U+ R \leq N_0$.
\end{thm}

\begin{proof} Let $Y = (S, E, I_{a}, I_{s}, U, R)$.
Thanks to the comparison principle and according to the Duhamel formulation, we look for a time $T>0$ such that the
map
\begin{eqnarray*}
  \Phi(Y) &:=& \left(
\G_d * S_0 - \int_0^t \G_d * \left( \omega \left( \beta_e E + \beta_s I_s + \beta_a I_a \right) \frac{S}{N} \right) ds, \right.
\\
&& 
\K_e * E_0 + \int_0^t \K_e * \left(  \omega \left( \beta_e E + \beta_s I_s + \beta_a I_a \right) \frac{S}{N}  \right)ds,
\\
&& 
\K_s * I_{a,0} +  (1-p) \delta \int_0^t \K_s *  E ds,
\\ 
&&
\e^{-(\gamma+\mu+\nu)t} I_0 + p \delta \int_0^t \e^{-(\gamma+\mu)(t-s)}  E ds,
\\ 
&&
\e^{-(\gamma+\mu)t} U_0 + \nu \int_0^t \e^{-(\gamma+\mu)(t-s)}  I_s ds,
\\ 
&&
\left. R_0 + \gamma \int_0^t  I_s +  I_a + U ds \right)
\end{eqnarray*}
is a contraction mapping from the closed ball
$$
\left\{ Y = (S, E, I_{a}, I_{s}, U, R) \in L^\infty(\R_+,L^\infty(\Omega))^6; \sup_{t\in [0,T]} || Y(t,.) - Y_0||_{L^\infty(\Omega)} < +\infty \right\}
$$
onto itself.
Here $\G_d, \K_e,$ and  $\K_s$ are the kernels of the respective operators 
$\partial_t - d\Delta$, $\partial_t - d\Delta + \delta $, and $\partial_t - d\Delta + \gamma$ for $d=d_0$ or $d_1$.
According to \cite{ouh05},  there exists a constant $C_\Omega>0$ depending only on $\Omega$ such that
the kernels satisfy
$$
||\G_d(t,.)||_{L^1(\Omega)} \leq C_\Omega, ||\K_e(t,.)||_{L^1(\Omega)} \leq C_\Omega,
||\K_s(t,.)||_{L^1(\Omega)} \leq C_\Omega.
$$
Combining with the fact that the integral terms of the right-hand-side are locally Lipschitz,
choosing $T\ll \frac{1}{C_\Omega}$ allows to apply Picard's fixed point theorem \citep{hen81}.

If $f=(f_1, \dots,f_6)$ denotes the right-hand-side of the system \eqref{e1}, and
$Y=(S, E, I_{a}, I_{s}, U, R)$, since $f_i(Y_i=0)\geq 0$, 
  we deduce that the solution is nonnegative if the initial datum is nonnegative.
 Finally, maximum principle provides boundedness of solution.
\end{proof}

We give a condition on parameters such that the disease has an exponential initial growth.

\begin{thm}
  Let $(S_0, E_0, I_{a,0}, I_{s,0}, 0,0)$ be a nonnegative initial datum.
If the basic reproduction number
$$
\mathcal{R}_0 := \omega_0 \left( \frac{ \beta_e  }{\delta} + \frac{ (1-p) \beta_a  }{ \gamma } + \frac{ p \beta_s  }{ \gamma + \mu + \nu }  \right) \frac{S_0}{N_0} > 1,
$$
then $(E,I_a,I_s)$ exponentially grows.
\end{thm}

This number has an epidemiological meaning. The term $ \frac{   \beta_e  }{ \delta  }$ represents the transmission rate by exposed during the average latency period  $1/\delta$. The term $\frac{  (1- p) \beta_a  }{ \gamma }$ is the transmission rate by asymptomatic during the average infection period $1/\gamma$, and the last one is the part of symptomatic.

\begin{proof}
A linearization around $(S_0, E_0, I_{a,0}, I_{s,0},0,0)$ is written as
the linear system of differential equations
\[
  \begin{pmatrix}
    E \\ I_a \\ I_s
  \end{pmatrix}'(t) =   \begin{pmatrix}   d_0\Delta + \omega_0\beta_e\frac{S_0}{N_0} -\delta   &  \omega \beta_a\frac{S_0}{N_0} &   \omega\beta_s\frac{S_0}{N_0}
            \\ (1-p) \delta & d_0\Delta -\gamma & 0
            \\ p \delta & 0 & -(\gamma+\mu+\nu) 
     \end{pmatrix}   \begin{pmatrix}
    E \\ I_a \\ I_s
  \end{pmatrix}. 
\]
Let $(v_k)_{k\geq 1}$ be an orthonormal basis of eigenfunctions of the Laplace operator with the homogeneous Neumann boundary condition, {\it i.e.} 
$-\Delta v_k = k^2 v_k$. Therefore, the characteristic polynomial of the matrix 
$$ \begin{pmatrix}   -d_0k^2 + \omega_0\beta_e\frac{S_0}{N_0} -\delta  & \omega  \beta_a\frac{S_0}{N_0}  &  \omega \beta_s\frac{S_0}{N_0}
            \\ (1-p) \delta & -d_0k^2 -\gamma & 0
            \\ p \delta & 0 & -(\gamma+\mu+\nu) 
     \end{pmatrix},$$
is $P(x) = x^3 + a_2 x^2 + a_1 x + a_0$,  
with $a_0 =  (\gamma+\mu)(d_0k^2 + \gamma) (d_0k^2+\delta)  \left( 1- \mathcal{R}_k \right)$ and 
$$
\mathcal{R}_k := \omega_0 \left( \frac{ \beta_e  }{d_0k^2+\delta} + \frac{(1- p) \beta_a  }{d_0k^2 + \gamma } + \frac{ p \beta_s  }{ \gamma+ \mu+\nu }  \right) \frac{S_0}{N_0}.
$$
If $\mathcal{R}_k>1$, there is at least one positive eigenvalue that coincides with an initial exponential growth rate of solutions.
\end{proof}
To reflect the spatiotemporal dynamic of the disease, we consider the effective reproduction number 
$$
\mathcal{R}_{\rm eff}(\x,t) := \omega(t) \left( \frac{ \beta_e  }{\delta} + \frac{ (1-p) \beta_a  }{ \gamma } + \frac{ p \beta_s  }{ \gamma + \mu + \nu }  \right) \frac{S(\x,t)}{N(\x,t)},
$$
and 
its mean with respect to the domain $\Omega$
$$
\overline{\mathcal{R}_{\rm eff}}(t) := \frac{1}{\mathcal{A}(\Omega)}\int_\Omega \mathcal{R}_{\rm eff}(\x,t) d\x.
$$
The value of $\mathcal{R}_0$ is computed in Table A of Figure \ref{fig1}.

We now establish the asymptotic behavior of solution.
\begin{thm} With the same assumptions as Theorem \ref{thm1}. Suppose moreover $\omega_0\beta_e \leq \delta$. Then the solution converges almost everywhere to the Disease Free Equilibrium $(S^*, 0, 0, 0, R^*)$ with $S^*+R^*=N^*$.
\end{thm} 

\begin{proof}
  From the last differential equation in system \eqref{e1}, we deduce that
$R$ is an increasing function bounded by $N(0)$. Thus $R(t)$ converges to $R^*$ a.e. as $t$ goes to $+\infty$.
Then integrating over time this equation provides
$$
    R(\x,t) - R(\x,0) = \gamma \int_0^t  I_a(\x,s) +  I_s(\x,s)  +  U(\x,s)  \, ds
$$
and
$$
    R^*(\x) - R_0(\x) =  \gamma\int_0^{+\infty} I_a(\x,s) +  I_s(\x,s)  +  U(\x,s) \, ds, 
$$
which is finite. Furthermore, $I_s, I_a, U$ also go to $0$ a.e. as  $t\to +\infty$
thanks to the positivity of the solution.
Multiplying the second equation by $E$ and integrating over $\Omega$ give
\[
  \frac{1}{2}\frac{d}{dt} \int_\Omega E^2 (\x,t) d\x + d \int_\Omega (\nabla E)^2 (\x,t) d\x = 
  \int_\Omega  \omega\left( \beta_s I_s + \beta_a I_a \right) \frac{S}{N} E + ( \omega\beta_e  \frac{S}{N} - \delta) E^2.
\]
Since $0\leq \frac{S}{N}\leq 1$ and $\omega\leq \omega_0$, Young's inequality followed by Poincar\'e inequality provide
\[
  \frac{1}{2}\frac{d}{dt} \int_\Omega E^2 (\x,t) d\x + d C_\Omega \int_\Omega E^2 (\x,t) d\x \leq 
  \int_\Omega \frac{\omega_0^2\beta_s^2}{2\varepsilon} I_s^2 + \frac{\omega_0^2\beta_a^2}{2\varepsilon} I_a^2
  + (\frac{\varepsilon}{2} + \omega_0 \beta_e   - \delta) E^2.
\]
Since $I_s$ and $I_a$ go to $0$ when $t\to +\infty$, it is enough to choose $\varepsilon>0$ such that
$dC_\Omega + \delta - \omega_0 \beta_e - \frac{\varepsilon}{2} > 0$, to conclude that $E \to 0$ a.e.
\end{proof} 

\begin{rem}
The basic reproduction number $\mathcal{R}_0$ can be computed thanks to the next generation matrix of the model without diffusion as in \cite{vdd2000}.
Since the infected individuals are in $E, I_a$ and $I_s$, 
 new infections ($\mathcal{F}$)  and  transitions between compartments
($\mathcal{V}$) can be rewritten  as 
$$\mathcal{F} =\begin{pmatrix} 
  \omega ( \beta_e E + \beta_s I_s + \beta_a I_a ) \frac{S}{N} \\
0 \\ 
0 
 \end{pmatrix}, 
\; \mathcal{V} = \begin{pmatrix} 
\delta_e  E \\
 \gamma I_a - (1-p) \delta E\\
(\gamma + \mu + \nu) I_s - p \delta E 
\end{pmatrix}.
$$
Thus, $\mathcal{R}_0 = \rho(FV^{-1})$ of the next generation matrix
  \begin{align*}
 FV^{-1} &= \begin{pmatrix} 
 \frac{  \omega \beta_e S_0 }{ \delta  N_0 } + \frac{\omega (1-p)  \beta_a S_0 }{ \gamma N_0} + \frac{\omega p   \beta_s S_0 }{ (\gamma + \mu+\nu) N_0 } &  \frac{ \omega \beta_a S_0 }{ \gamma N_0}  &  \frac{\omega   \beta_s S_0 }{ (\gamma + \mu +\nu ) N_0 } \\
  0 & 0 & 0 \\
  0 & 0 & 0
 \end{pmatrix}.
\end{align*}
\end{rem}

\subsection{Model resolution}

No treatment is applied, then $\nu=0$ 
Calibration of the model is done from January 24, 2020, the day of first confirmed infection, to April 30, {\it i.e.} 97 days.
Since $\mathcal{R}_0=3.425257$ and $1-\frac{1}{\mathcal{R}_0}=0.708051$,  $70\%$ of the population is set as susceptible to the infection due to the
herd immunity.  
The objective function $J$ is computed to provide a relative error of order less than $10^{-2}$.
In Figure \ref{fig1}, Table A shows estimated parameters. Remark that $\omega_0\beta_e \leq \delta$.
The rest of the Figure compares the data and the fitted total symptomatic infected and death of the posterior distribution.


\begin{figure}[htbp]
\centering
  \includegraphics[scale=.42]{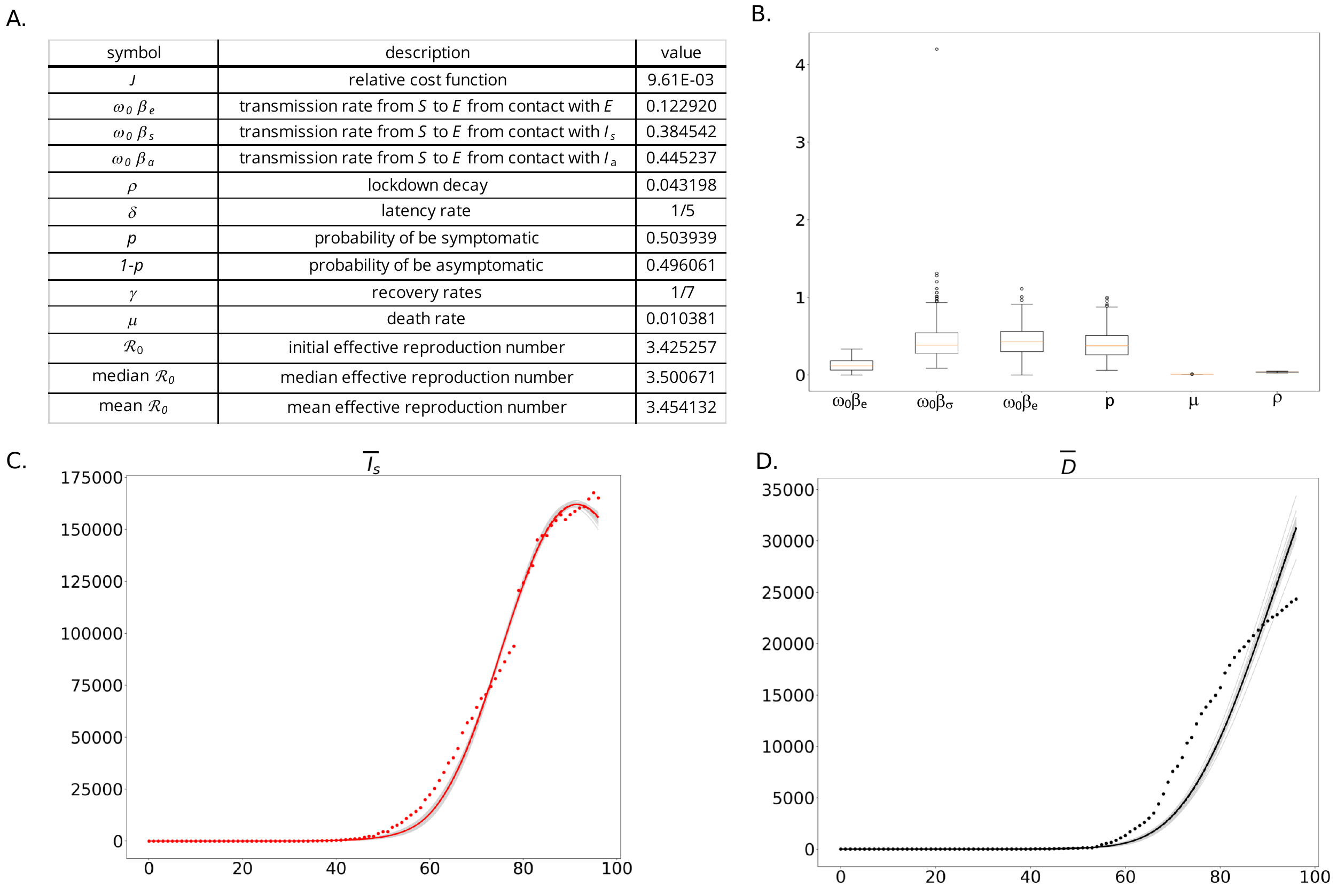}
\\[-0.5cm]\caption{A. Parameters calibrated according to data  from Sant\'e Publique France. B. Boxplot of the posterior distribution computed from these data. C. Fitted total symptomatic infected of the posterior distribution in grey , median in red straight line, mean is dotted line. D. Fitted total death  of the posterior distribution in grey, median in black straight line, mean is dotted line.}
\label{fig1}
\end{figure}
\newpage

\subsection{Spatial spread of covid-19}

Simulations are performed from  January 24, 2020, the day of first confirmed infection, to June 16.
Images are $957\times 984$ pixels, time step is chosen to verify the stability condition preserving positivity,
and $\eta = 0.6$.

Figure \ref{fig2}-A presents the day before lockdown, the disease is mainly located in regions North (Hauts-de-France, \^Ile-de-France), ,  East (Alsace), Center-East (Rh\^one) and South-East (Bouches-du-Rh\^one).
When lockdown ends on May 10, Figure \ref{fig2}-B1 shows that  the disease remains located into these regions,
 except for the South where it almost vanishes. Provided that distancing is predominantly respected (here $\eta = 0.6$), the density of symptomatic infecteds is decreasing until June 16 everywhere (Figure \ref{fig2}-B2).

Figures \ref{fig2}-C represents what would happen if nothing has been done. 
On May 10, the disease is located in the same regions, but with a greater number of infected.
As a consequence, the spread continues from East to West. While the disease vanishes in the Eastern part of the country 
on June 16, the West is peaking.
It is important to see that, without intervention, the total number of infected is high and the whole country is deeply affected.

\begin{landscape}
\begin{figure}[htbp]
\centering
 \includegraphics[scale=.6]{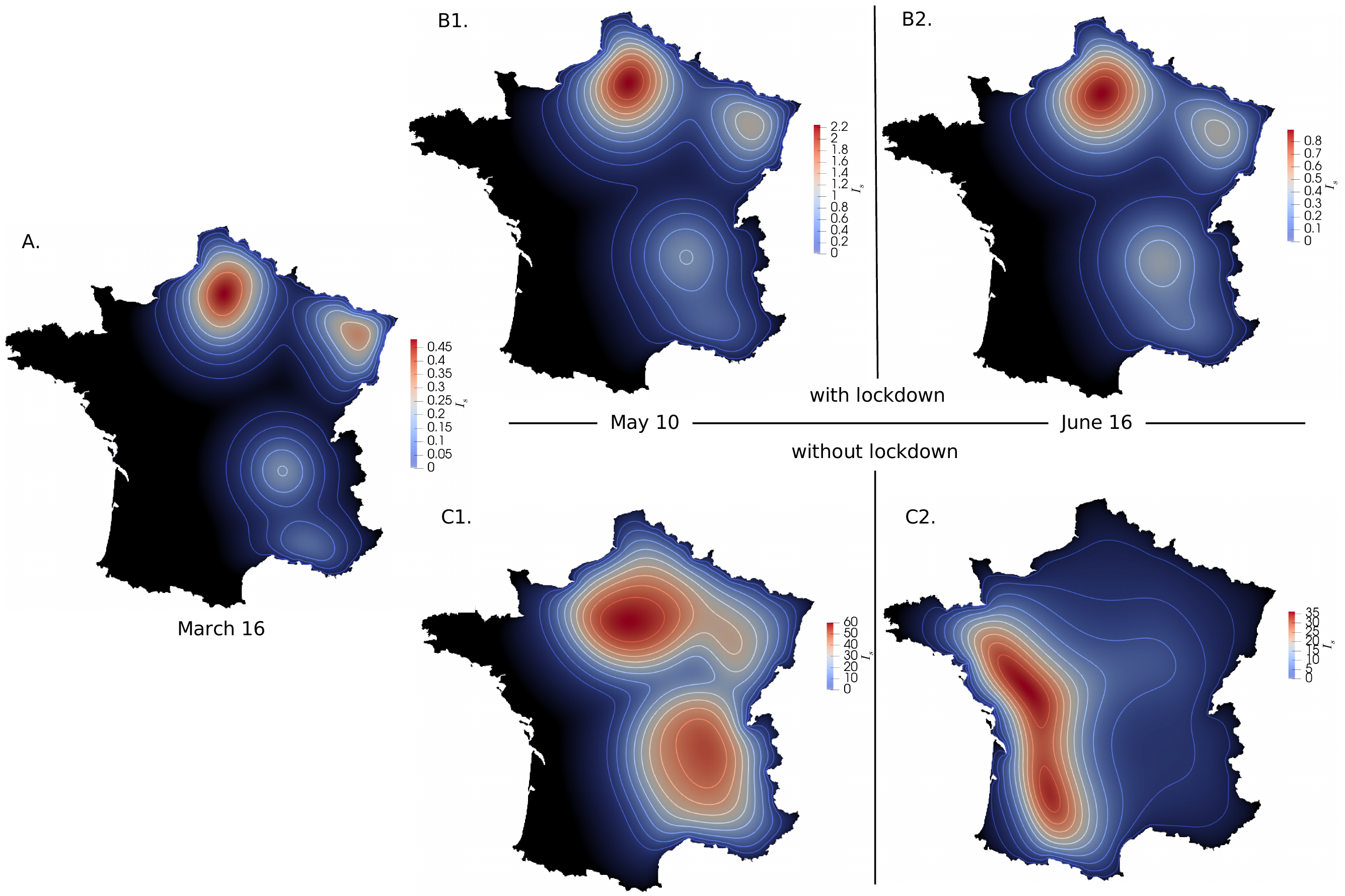}
 \\[-0.5cm]\caption{A. Spatial distribution of symptomatic infected density on March 16, one day before lockdown. 
 B. Spatial distribution of symptomatic infected density on May 10 (B1), one day before unlockdown, and on June 16 (B2). 
 C. Spatial distribution of symptomatic infected density on May 10 (B1) and on June 16 (B2) assuming that no lockdown has been done. }
\label{fig2}
\end{figure}

\end{landscape}

\subsection{Strategies for unlockdown}

A naive strategy is to unlock the regions where $\mathcal{R}_{\rm eff}(\x,.)$ is the lowest.
Looking at the maps in Figures \ref{fig3}, we see that the effective reproduction number before the lockdown 
 is between $3.34$ and $3.42$ (Figures \ref{fig3}-A). On May 10, it goes down and is between $0.38$ and $0.42$ 
 (Figures \ref{fig3}-B).
After May 10, with a distancing mostly respected ($\eta = 0.6$), $\mathcal{R}_{\rm eff}(\x,.)$ increases again.
It exceeds $1$ with lower value ($1.04$) in the regions that have been most affected upstream (Figures \ref{fig3}-C).

\begin{figure}[htbp]
\centering
 \includegraphics[scale=.48]{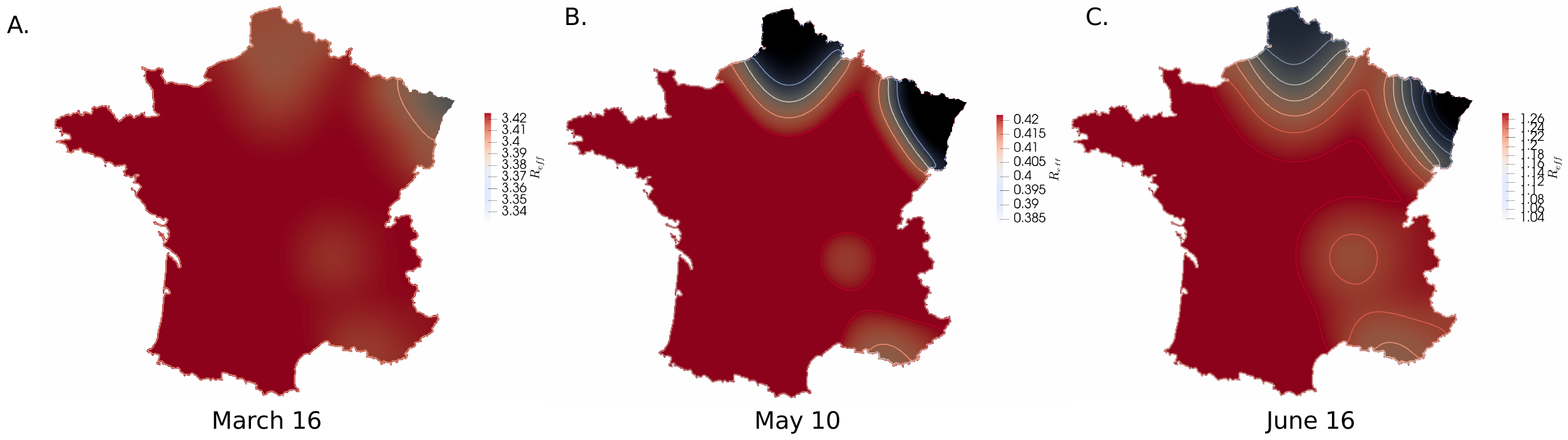}
 \\[-0.5cm]\caption{Spatial distribution of effective reproduction number one day before lockdown March 16 (A),
  one day before unlockdown May 10 (B), and on June 16 (C). }
\label{fig3}
\end{figure}



A less forced strategy is to believe in the respect of distancing after lockdown. With this set of parameters, 
$\overline{\mathcal{R}_{\rm eff}}$ goes to $1$ when $\eta = 0.63$.
As shown in Figure \ref{fig4}-A, provided that at least $63\%$ of the individuals respect the distancing rules,
the number of symptomatic infected individuals is controlled. Figure \ref{fig4}-B shows that in this case, the
effective reproduction number $\overline{\mathcal{R}_{\rm eff}}$ remains less than $1$. 
Below $63\%$ regarding distancing, the number of symptomatic infected increases again and 
$\overline{\mathcal{R}_{\rm eff}}$  becomes greater than $1$. 

A third strategy, more constraining for some, is to continue the lockdown in the most affected areas.
Here, lockdown is continued on the eastern half, while the western half is free.
Then the number of symptomatic infected and the effective reproduction number 
pursue their decay as shown in Figure \ref{fig4} (dotted  lines).
It is important to note that in the simulation no individuals move 
from a confined to an unconfined region.  A homogeneous Neumann condition, equal to $0$, is
imposed along this fictitious frontier.

\begin{figure}[htbp]
\centering
 \includegraphics[scale=.48]{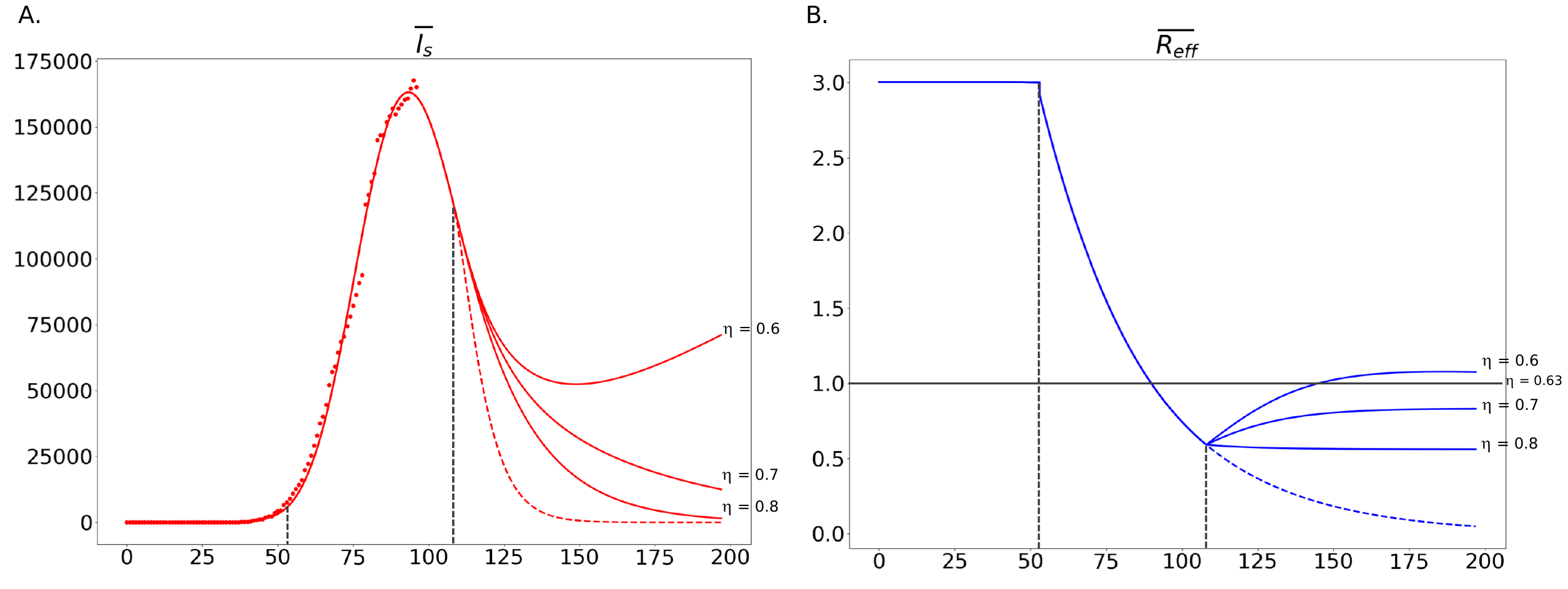}
 \\[-0.5cm]\caption{A. Total number of symptomatic infected with $\eta = 0.6,0.7,0.8$ (\textcolor{red}{-----}) and with lockdown on Eastern region (\textcolor{red}{- - -}). B. Mean effective reproduction number with $\eta = 0.6,0.7,0.8$ (\textcolor{blue}{-----}) and with lockdown on Eastern region (\textcolor{blue}{- - -}). Vertical lines represent the day of lockdown (March 17) and of unlockdown (May 11). Dots denote for the confirmed data.}
\label{fig4}
\end{figure}


\newpage

\section{Discussion}

Lockdown has reduced both the number of infections and the spread.
The propagation focused in the Eastern half and kept the West intact.
Without lockdown, the whole country would be affected and more severely.
Lockdown caused a significant reduction of the reproduction number, from $3.42$ to $0.38$. 
Since the number of susceptible individuals remain large, as soon as contacts increase, the effective reproduction number
$\mathcal{R}_{\rm eff}$ grows and can exceed $1$.

With a  lack of treatment, social distancing remains the most effective means. We notice that it has to
be highly respected (here at over $63\%$). As shown in Figure \ref{fig4}-A, 
the number of symptomatic infected individuals, therefore potentially hospitalized, is restrained as soon as at least $63\%$ 
of the individuals respect the distancing,
Nevertheless, this constraint can be relaxed since it can be imposed only in the most infectious areas.

In summary, to obtain a possible unlockdown map, the local value of the effective reproduction number 
should be taken into account, as well as the number of infected individuals and the direction of the spread of the disease.

Of course, this is a simplified model based only the population density and mean daily commuting.
 For example, the model could be improved by considering a larger diffusion  
 along major axes of travel ({\it i.e.} $d=d(\x,t)$), by taking local effects of distancing ({\it e.g} $\eta = \eta(\x,t)$ where $\eta(\x,t)=0$ in closed schools),
 or by opening of the frontiers (by changing the boundary conditions and adding new recruits).

\noindent
{\bf Supplementary video} Spatiotemporal propagation of COVID-19 from March 16 to June 16
with lockdown occurring from March 17 to May 11 (left) and without intervention (right).

\bibliographystyle{abbrvnat}
{\small
\bibliography{references}
}


\end{document}